\documentclass{article}
\usepackage{xcolor,graphicx} % Required for inserting images
\usepackage{amsmath,amsthm,amssymb,amsfonts,bm}
\usepackage{fullpage,physics}
\usepackage[bookmarksnumbered,colorlinks,citecolor=red,linkcolor=blue,hyperindex,linktocpage=true]{hyperref}
\usepackage[nameinlink,capitalise]{cleveref}
\usepackage{mathtools, array}
\usepackage{scalerel}
\usepackage{stackengine,wasysym}

\newtheorem{theorem}{Theorem}[section]
\newtheorem{corollary}{Corollary}[theorem]
\newtheorem{lemma}[theorem]{Lemma}
\newtheorem{assumption}{Assumption}{\bf}{\it}
\newtheorem{remark}[theorem]{Remark}
\newtheorem{definition}[theorem]{Definition}

\newtheorem{example}[theorem]{Example}
\newtheorem{assumptions}[theorem]{Assumption}

\title{Structure-preserving quantum algorithms for linear and nonlinear Hamiltonian systems }

\author{Hsuan-Cheng Wu\thanks{wu.hsuancheng@psu.edu} \\
Xiantao Li\thanks{Xiantao.Li@psu.edu}\\
Department of Mathematics, \\ The Pennsylvania State University,\\  University Park, Pennsylvania 16802, USA
}
\def\R{\mathbb{R}}
\def\N{\mathbb{N}}

\def\e{\epsilon}
\def\l{\lambda}

\def\iff{\Longleftrightarrow}

\def\<{\langle}
\def\>{\rangle}

\newcommand\reallywidetilde[1]{\ThisStyle{%
  \setbox0=\hbox{$\SavedStyle#1$}%
  \stackengine{-.1\LMpt}{$\SavedStyle#1$}{%
    \stretchto{\scaleto{\SavedStyle\mkern.2mu\AC}{.5150\wd0}}{.6\ht0}%
  }{O}{c}{F}{T}{S}%
}}

\def\tilde#1{%
  \reallywidetilde{#1}
}

\newcommand{\mat}[1]{\begin{pmatrix}#1 \\ \end{pmatrix}}

\begin{document}

\maketitle

\begin{abstract}
    Hamiltonian systems of ordinary and partial differential equations are fundamental mathematical models spanning virtually all physical scales. A critical property for the robustness and stability of computational methods in such systems is the underlying symplectic structure, which preserves geometric properties like phase-space volume over time and energy conservation over an extended period. In this paper, we present quantum algorithms that incorporate symplectic integrators, ensuring the preservation of this key structure. We demonstrate how these algorithms maintain the symplectic properties for both linear and nonlinear Hamiltonian systems. Additionally, we provide a comprehensive theoretical analysis of the computational complexity, showing that our approach offers both accuracy and improved efficiency over classical algorithms. These results highlight the potential application of quantum algorithms for solving large-scale Hamiltonian systems while preserving essential physical properties.
\end{abstract}

\section{Introductions}

Hamiltonian systems are foundational across various physical scales, from electron dynamics to the motion of atoms and molecules, and even to macroscopic models in continuum mechanics \cite{meyer2009hamiltonian}. These systems are central to modern statistical mechanics \cite{Balescu76,Toda-Kubo-2}. Moreover, Hamiltonian systems encompass numerous well-known examples across diverse scientific disciplines: from the Schrödinger equation governing electron dynamics in quantum mechanics to the cubic Schrödinger equation describing Bose-Einstein condensation; from Newton's equations of motion in classical mechanics to Euler's equation and the Korteweg–de Vries equation in fluid mechanics. Furthermore, Hamiltonian dynamics finds applications in elastodynamics within solid mechanics, Lotka-Volterra equations modeling population dynamics, and even the complex motions of celestial bodies. This ubiquity highlights the pivotal role of Hamiltonian dynamics in understanding natural phenomena across multiple scales and fields of study. Further, the underlying variational and geometric structures have given rise to significant advancements in mathematical theory \cite{cabral2023normal}.

Mathematically, Hamiltonian systems are governed by the system of ODEs, involving generalized coordinates and momenta $\bigl(\bm q(t), \bm p(t) \bigr) \in \mathbb{R}^{2d}$
\begin{equation}\label{eq: ham-ode}
    \left\{
    \begin{aligned}
         \frac{d}{dt}\bm q=& \nabla_{\bm p} H, \\
         \frac{d}{dt}\bm p= & - \nabla_{\bm q} H,
    \end{aligned} \right.
\end{equation}
where $H(\bm q, \bm p)$ is known as the Hamiltonian. The ODE system in \cref{eq: ham-ode} can be written compactly as $\bm x'= -J \nabla H(\bm x)$ with $\bm x$ incorporating both sets of variables \cite{hairer2006geometric}. 

The transition from a Hamiltonian $H$ to the ordinary differential equations (ODEs) in \cref{eq: ham-ode} embodies the celebrated Hamilton's principle, equivalent to the fundamental least-action principle. This elegant formulation has been extended to infinite-dimensional systems, including partial differential equations (PDEs),  with many further applications.

The simulation of large-scale Hamiltonian systems remains a critical focus in scientific computing. A significant advancement in numerical algorithms is the development of symplectic integrators, which preserve the underlying geometric structure of these systems (see \cref{def-symp}). Notably, symplectic integrators maintain energy accuracy over extended time period, a crucial feature for long-term simulations. In contrast, conventional methods such as Taylor approximations and explicit Runge-Kutta schemes often lead to energy drift, compromising both the physical fidelity and long-term accuracy of simulations. Thus the development of symplectic integrators has become a milestone in modern computational mathematics \cite{feng2010symplectic,marsden2001discrete,hairer2006geometric}.  To illustrate the significance of symplectic integrators, \cref{fig:energy} presents a comparative analysis of total energy conservation in a large-scale particle system. The simulation, conducted over 100,000 time steps, contrasts the performance of a symplectic integrator (Verlet's method) against a standard Runge-Kutta method. The results demonstrate the superior energy conservation properties of the symplectic approach.
\begin{figure}[htp]
    \centering
    \includegraphics[width=0.5\linewidth]{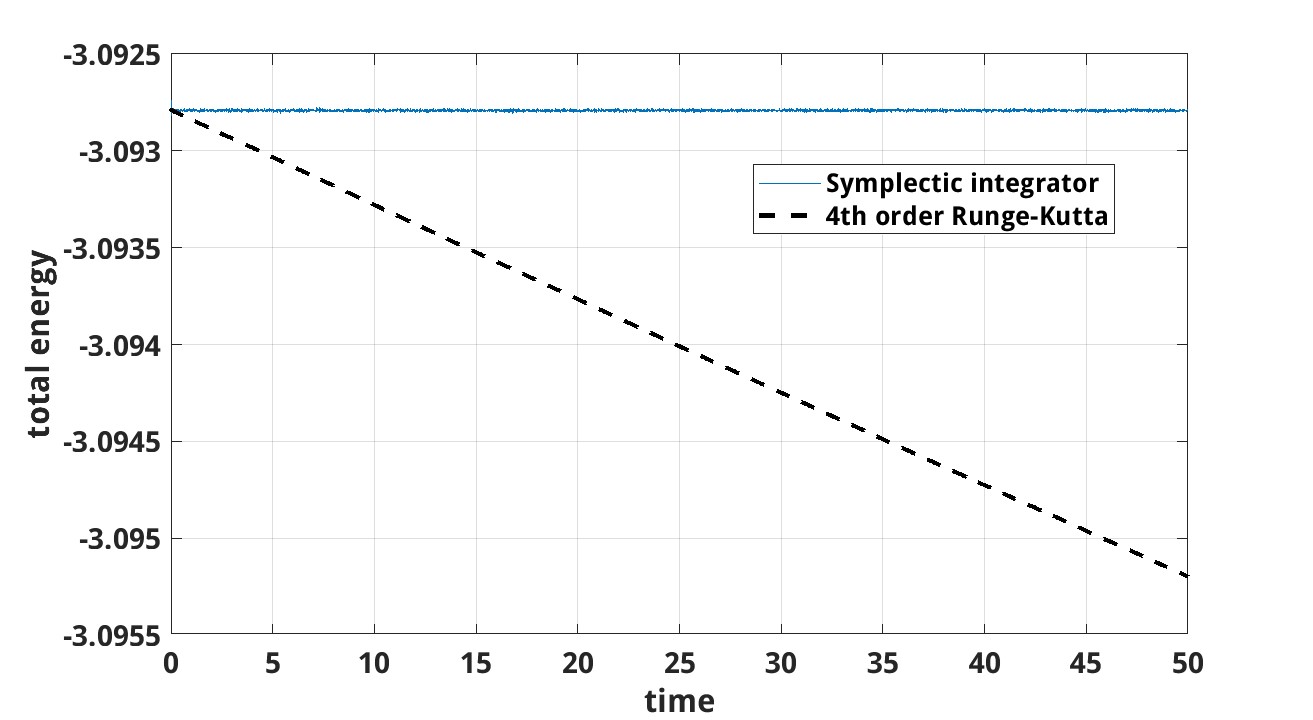}
    \caption{Total energy of a Lennard-Jones system with 2048 particles integrated in time for 100,000 steps.}
    \label{fig:energy}
\end{figure}

Quantum algorithms have shown remarkable potential for efficiently simulating high-dimensional systems. Notably, Hamiltonian simulation algorithms can simulate the unitary dynamics of the Schr\"odinger equation with exponential speedup relative to system dimension for many important cases. This quantum advantage has been extended to classical systems of harmonic oscillators \cite{BBK23} and the acoustic wave equation \cite{Costa2019Wave}, which are also Hamiltonian systems. Furthermore, quantum algorithms have been developed for simulating both linear and nonlinear ordinary differential equations (ODEs), see  \cite{an2023linear,berry2017quantum,berry2022quantum,childs2020quantum,fang2023time,jin2022quantumshort,krovi2023improved,liu2021efficient,wu2024quantum,succi2023quantum,tennie2025quantum,jin2022quantum,brearley2024quantum,joseph2020koopman} and the references therein.

Meanwhile, existing quantum algorithms have not explicitly addressed the symplectic structure inherent to Hamiltonian systems. Most of them still rely on Taylor approximation and simple Runge Kutta methods.  This paper addresses this gap by incorporating symplectic integrators into quantum algorithms and analyzing how the geometric structure impacts computational complexity. It is important to emphasize that 
while existing works have largely concentrated on optimizing the query and gate complexity of quantum algorithms, preserving the fundamental physical properties—similar to the evolution of classical simulation algorithms—remains equally important. When resources are insufficient to reach the desired precision, the robustness of these algorithms becomes increasingly critical in ensuring reliable outcomes.

\begin{figure}{thp}
    \centering
    \includegraphics[width=0.6\linewidth]{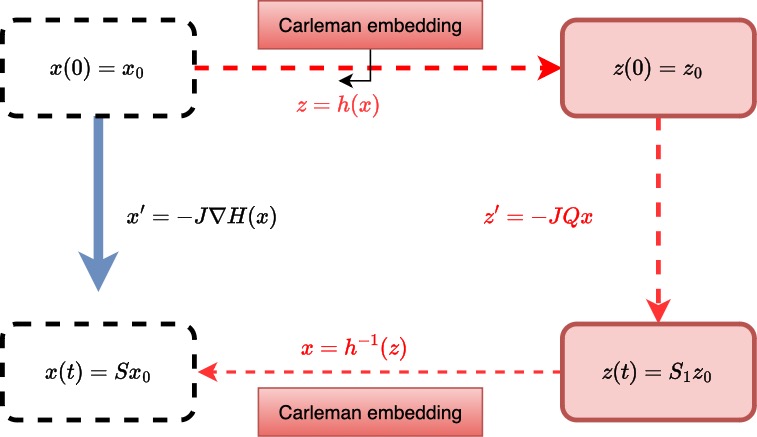}
    \caption{A schematic illustration of the Carleman embedding for Hamiltonian systems.}
    \label{fig:enter-label}
\end{figure}

For linear Hamiltonian dynamics, we incorporate the Runge-Kutta Gauss methods \cite{deuflhard2012scientific} which automatically satisfies the symplectic property, while having the optimal order of accuracy for a fixed stage number.
\begin{theorem}Assume that $\bm x(t)$ is the solution of a linear stable Hamiltonian system, expressed compactly as $\bm x'= K\bm x$.  
    The quantum algorithm produces a quantum state $\ket{\psi} = \frac{\bm x_M}{\norm{\bm x_M}}   $, in which $ \norm{\bm x_M - \bm x(T)}< \epsilon. $ The mapping from the initial value $\bm x(0)$ to $\bm x_M$ is symplectic. The algorithm involves   $ \Tilde{O}(T \|K\| \kappa(V)^2 )$ queries to $K$, where in $\Tilde{O}$ we neglected logarithmic factors and $\kappa(V)$ is the condition number associated with the eigenvector matrix of $K$. 
\end{theorem}

Meanwhile, 
at a high level, our algorithm for nonlinear dynamics is outlined in \cref{fig:enter-label}. Specifically, the solution of a Hamiltonian system corresponds to a symplectic map: $\bm x(t)= S \bm  x(0).$ We will show that the Carleman embedding implicitly induces a \emph{nonlinear} symplectic map $\bm z=\bm h(\bm x)$ that transforms the nonlinear dynamics to a \emph{linear} Hamiltonian system for which the symplectic structure remains and our algorithm in the previous theorem is applicable. We denote the symplectic map of the linear dynamics by $S_1$. Therefore, the original symplectic solution map is factored as follows,
\begin{equation}
    S = \bm h^{-1} \circ S_1  \circ \bm h,
\end{equation}
each of which is symplectic, a property that is invariant under composition.  Our symplectic integrators, when applied to the Carleman system, introduce an approximate, but still symplectic, map for $S_1.$ This result is summarized in the following.
\begin{theorem}[Informal Version of \cref{thm2}] Assume that the $n$ dimensional nonlinear Hamiltonian system satisfies a no-resonance condition and a condition on the strength of the nonlinearity $ \mathsf{R}_r <1$ (see \cref{assump1}). For  $T\le \frac{1}{e\norm{F_2}\kappa(V)^N}$,
 there is a quantum algorithm that produces a quantum state $\ket{\psi} = \frac{\bm x_M}{\norm{\bm x_M}}   $, in which $ \norm{\bm x_M - \bm x(T)}< \epsilon. $  The algorithm involves $ \Tilde{O}\left( \frac{T^{1+2\log ( \kappa(V))  }  }{\epsilon^{2\log (\kappa(V))  }} (\norm{F_1} + \norm{F_2}) \right)$ queries to the coefficient matrices in the Hamiltonian. The mapping $S$ from the initial value $\bm x(0)$ to $\bm x_M$ is approximately symplectic in the sense that
    \[  
    S^T J S - J = O(\epsilon). 
    \]
\end{theorem}
It is important to emphasize that the commonly cited term ``Carleman linearization" is somewhat misleading: while the goal is to reduce the problem to a linear dynamical system, the underlying transformation itself is highly nonlinear. Therefore, we prefer the terminology ``Carleman embedding".

\paragraph{Main contributions}
\begin{itemize}
    \item We present a quantum algorithm for linear Hamiltonian systems with complexity that is optimal in both $T$ and $\epsilon$, while preserving the symplectic properties of the solution map. As a result, energy conservation is exactly maintained.
    \item Our algorithm implements the symplectic collocation methods, which in practice typically outperforms explicit Runge Kutta methods that involve the same number of states.
    \item Unlike the complexity estimates for general ODEs \cite{fang2023time,krovi2023improved}, our result does not depend on any time-varying parameters.
    \item We extend this structure-preserving property to nonlinear Hamiltonian systems with polynomial nonlinearity by utilizing the Carleman embedding technique, and quantify how the symplectic properties are affected by the Carleman embedding. 
    \item We analyze the complexity of the quantum algorithm for the finite Carleman system. 
    \item We demonstrate that the Carleman embedding implicitly induces a symplectic map, which, under our conditions, reduces the nonlinear Hamiltonian dynamics to a linear Hamiltonian system. Consequently, our quantum algorithm preserves the symplectic property up to an error $\epsilon$, which arises from the finite truncation of the Carleman system.
\end{itemize}

\paragraph{Related works}

{Numerous efficient quantum algorithms have been developed for solving linear ODE systems \cite{an2023linear,berry2017quantum,berry2022quantum,childs2020quantum,fang2023time,jin2022quantumshort,krovi2023improved}, with many of these algorithms reducing the problem to a Hamiltonian simulation framework—solving the unitary dynamics of a Schrödinger equation. These algorithms have been applied to harmonic oscillators \cite{BBK23} and the acoustic wave equation \cite{Costa2019Wave}, which are two examples of classical Hamiltonian systems. Quantum algorithms in this framework offer exponential speedups with respect to problem dimensions.  The recent work \cite{krovi2024quantum} has specifically considered Hamiltonian systems. However, quantum algorithms designed to handle Hamiltonian ODEs that preserve symplectic structures have not yet been explored in these prior works. 

In contrast, nonlinear Hamiltonian systems pose significantly more challenges due to their inherent complexity and interactions. One regime where quantum advantage has been identified is in dissipative dynamics using Carleman embedding \cite{carleman1932application}, where the real part of the Jacobian’s eigenvalues is strictly negative \cite{liu2021efficient}. Hamiltonian systems, due to their time-reversible nature, fall outside this regime. However, the authors' recent work \cite{wu2024quantum} bridges this gap by identifying new conditions under which quantum advantage remains achievable. This new condition will be used in the current paper as well. 

For general nonlinear dynamical systems, Br\"ustle and Wiebe \cite{brustle2024quantum} provided a lower bound on the computational complexity, showing that it scales exponentially with the time duration $T$. They also leveraged Carleman embedding to develop algorithms with the same complexity scaling.
}

\section{Preliminaries}\label{sec:pre}

This paper is concerned with the time integration in a time interval $[0,T]$. For the time discretization, we denote the step size by $\tau >0,$ and the time steps $t_n=n \tau, n \in [M]$, with $M=T/\tau.$ The corresponding approximation solution of $\bm x(t_n)$ is denoted by $\bm x_n.$

We use bold fonts for vectors, e.g., $\bm v$, and use $\norm{\bm v}$ to denote its Euclidean norm. Unless specified otherwise, $\norm{A}$ denotes the norm of a matrix $A$  induced by the vector norm.  We also use $\kappa(A):= \norm{A} \norm{A^{-1}}.$ In addition, $\norm{\bm v}_1$ refers to the $\ell_1$ vector norm, with the induced matrix denoted by   $\norm{A}_1$.

In addition, we provide the definition of symplectic mappings from \cite{hairer2006geometric}.
\begin{definition}\label{def-symp}
    A linear mapping $A:\R^{2d}\longrightarrow\R^{2d}$ is called symplectic if
    \[ A^TJA=J,  \]
    where 
    \begin{equation}\label{Jmat}
         J\coloneqq\mat{O& I_d\\ -I_d & O}. 
    \end{equation}
More generally, a differentiable map $f:\R^{2d}\longrightarrow\R^{2d}$ is called symplectic if the Jacobian matrix of $f$ is symplectic.
\end{definition}

\begin{definition}
    A numerical method, applied to any Hamiltonian system is called symplectic if the mapping from the initial condition, e.g., $\bm x_0$, to solutions at any later step, e.g., $\bm x_n$, is symplectic.
\end{definition}

In our algorithm, we will use block encoding as the input model. In particular, they are the building blocks for assembling large matrices and implementing matrix operations. 
Below we provide some known definitions and lemmas
of block encoding for the proposed implementing quantum algorithms.
\begin{definition}
    Let $A$ be an $n-$qubit operator, $\alpha,\epsilon>0$ and $m\in\N$. An $m+n$ qubit unitary $U$ is called an $(\alpha,m,\epsilon)$ block encoding of $A$ if
    \[ \|A-\alpha(\bra{0}^m\otimes I)U(\ket{0}^m\otimes I)\|\le\epsilon. \]
\end{definition}

\begin{lemma}[{\cite[Lemma 48]{GSLW19}}]
  \label{lemma:sparse-to-be}
  Let $A\in\mathbb{C}^{2^n\times 2^n}$ be an $n$-qubit operator with at most $s$ nonzero entries in each row and column. Suppose that $A$ is specified by the following sparse-access oracles:
  \begin{align}
    \label{eq:sparse-1}
    O_A: &\ket{i}\ket{j}\ket{0} \mapsto \ket{i}\ket{j}\ket{A(i, j)}, \text{ and}\\
    \label{eq:sparse-2}
    O_S: &\ket{i}\ket{k}\ket{0} \mapsto \ket{i}\ket{j}\ket{r_{i,k}},
  \end{align}
  where $r_{i,k}$ is the $k$-th nonzero entry of the $i$-th row of $A$. Suppose $|A_{i,j}| \leq 1$ for $i \in [m]$ and $j\in[n]$. Then for all $\epsilon \in (0, 1)$, an $(s, n+3, \epsilon)$-block-encoding of $A$ can be implemented using $O(1)$ queries to $O_A$ and $O_S$, along with $O(n+\mathrm{polylog}(1/\epsilon))$ 1- and 2-qubit gates.
\end{lemma}

\begin{definition} \label{Hermitian Complement}
    The Hermitian complement of matrix $A$ is defined by
    \[ \bar{A}\coloneqq\mat{0 & A \\ A^\dag & 0}. \]
\end{definition}

\begin{lemma}[\cite{chakraborty2018power} Lemma 9, \cite{krovi2023improved} Remark 1] 
\label{lem:BE-inv}
Let $A\in \mathbb{C}^{2^n\times 2^n}$ be a matrix with condition number $\kappa(A)\ge2$. Suppose $\frac{I}{\kappa(A)}\preceq \bar{A}\preceq I$. For any fixed $\epsilon>0$, we define $\delta\coloneqq o\left(\frac{\varepsilon}{\kappa(A)^2\log^3(\frac{\kappa(A)^2}{\epsilon})}\right)$, then if $U$ is an $(\alpha,n,\delta)$ block encoding of $\bar{A}$ with gate complexity $T_U$, we can implement a 
\[ \left(2\kappa(A),n+O(\log(\kappa(A)^2\log\varepsilon^{-1})), \epsilon\right) \]
block encoding of $\bar{A}^{-1}$ with gate complexity 
\[ O\left(\alpha\kappa(A)(n+T_U)\log^2(\frac{\kappa(A)^2}{\epsilon})\right). \]
Furthermore, a block encoding of $A^{-1}$ can be obtained from the upper-right block of $U$, in that, 
\[ \norm{A^{-1}-2\kappa(A)(\bra{0}^m\otimes \bra{1}\otimes I)U(\ket{0}^m\otimes\ket{0}\otimes I)}\le\epsilon. \]
\end{lemma}

\begin{lemma}[Product and Linear Combinations of Unitaries \cite{childs2012hamiltonian}]
\label{lem:BE-lcu}
Let $A$ and $B$ have an $(\alpha,n,\epsilon)$ block encoding with gate complexity $T_A$ and a $(\beta,m,\delta)$ block encoding with gate complexity $T_B$, respectively. Then 
\begin{enumerate}
    \item[(a)] We can implement an $(\alpha,n+1,\epsilon)$ block encoding of $\bar{A}$ with gate complexity $O(T_A)$.
    \item[(b)] There exists an $(\alpha+\beta,n+m,\alpha\delta+\beta\epsilon)$ block encoding of $A+B$ with gate complexity $O(T_A+T_B)$.
    \item[(c)] There exists an $(\alpha\beta,n+m,\alpha\delta+\beta\epsilon)$ block encoding of $AB$  with gate complexity $O(T_A+T_B)$. 
\end{enumerate}
\end{lemma}

\begin{lemma}[\cite{krovi2023improved} Theorem 1]\label{complexity lemma}
\label{lem:qls}
Suppose $\|A\|=O(1)$ so $\|A^{-1}\|=O(\kappa(A))$. Suppose an oracle for block encoding of $A$ and an oracle for preparing $\ket{\bm b}$ are provided. There exists a quantum algorithm that outputs the normalized state $A^{-1}\ket{\bm b}$ within error $\epsilon$ by using $O(\kappa(A)\log\epsilon^{-1})$ calls to the oracles.   
\end{lemma}

\section{Symplectic Quantum Algorithms for Linear Hamiltonian Systems}

We first consider a linear Hamiltonian system with a quadratic Hamiltonian,
\begin{equation}
    H(\bm x)= \bm x^T Q \bm x, \quad Q\in \mathbb{R}^{2d\times 2d},\quad Q^T=Q, \quad \bm x=(\bm q, \bm p).
\end{equation}

In this case, the Hamiltonian ODEs in \cref{eq: ham-ode} can be written in the following compact form,
 \begin{equation}\label{linear ode}
     \frac{d}{dt} \bm x =K \bm x, \quad K:=J^{-1}Q.
 \end{equation}
Here $J$ is defined in \eqref{Jmat}.

\begin{assumptions}
    The Hamiltonian dynamics given in \cref{linear ode} is stable. In particular, there exists a constant $C$ such that $\norm{e^{tK }}< C$ for all $t\in \mathbb{R}$.
\end{assumptions}

\begin{assumptions} \label{sparsity K}
    The matrix $K$ is s-sparse. In other words, there are at most $s$ nonzero entries in each row and column of $K$, with sparse access described \cref{lemma:sparse-to-be}.
\end{assumptions}

The stability is directly related to the diagonalizability of the matrix $K$.
\begin{theorem}{\cite[Theorem 2.1]{meyer2013normally}}\label{as-stab}
    The Hamiltonian dynamics \cref{linear ode} is stable if and only if the matrix $K$ is diagonalizable with purely imaginary eigenvalues. 
\end{theorem}

Due to this stability property, we can express $K$ in a diagonalized form,
\begin{equation}\label{diagonalization K}
    K= V \Lambda V^{-1},
\end{equation}
and $\Lambda$ only has purely imaginary eigenvalues.  Therefore $e^{tK } = V e^{t\Lambda} V^{-1}$. Since $e^{t\Lambda}$ is unitary, we have a natural global 
bound,
\begin{equation}\label{stable K}
    \norm{e^{tK }} \leq \kappa(V), \; \forall t\in \mathbb{R}.
\end{equation}

Another possible characterization is through a symplectric map:  $K= B S B^{-1},$ where  
$B$ is symplectic and $S$ is skew-symmetric. In this case, we have $ \norm{e^{tK }} \leq \kappa(W). $

\subsection{The Runge-Kutta Gauss methods}

One class of ODE solvers that preserve the symplectic structure is the Runge-Kutta Gauss (RKG) methods, which are implicit methods with intermediate stages defined at the Gaussian quadrature points within each time interval. A key reason for choosing RKG methods is due to its optimal accuracy for a fixed stage number $p$ \cite{deuflhard2012scientific}. Another practical advantage is that RKGs are one-step methods, which, unlike multistep methods, do not require multiple initial starting values. Most importantly, any RKG method preserves the symplectic structure.

We first show the general steps of an $p$-stage RKG method, $p\in \mathbb{N}$. Suppose we have a $d$-dimensional ODE system $\bm x'=\bm f(t,\bm x)$ on space $\Omega$. We construct a polynomial function $\bm u(t) \in \mathbb{P}_{p}^d$ that collocates with $p$ given points given by Gauss-Lengendre quadrature points $\{c_j\}$ in $(0,1)$.  Specifically,  for each time step with step size $\tau$, to map the numerical solution $\bm x_n$ to $\bm x_{n+1}$, the method can be formulated as follows, find $\bm u(t) \in \mathbb{P}_{p}^d$ such that,
\begin{enumerate}
    \item[(1)] $\bm u(t)=\bm x_n$.
    \item[(2)] $\bm u'(t+c_i\tau)=\bm  f(t+c_i\tau,\bm u(t+g_i\tau))$, for $i=1,\cdots,p$.
\end{enumerate}
With $\bm u$ determined, one then sets  $\bm x_{n+1}= \bm u(t+\tau)$.
Those properties of the RKG that are relevant to the current paper are summarized below.
\begin{theorem}
    The RKG method, when applied to the linear Hamiltonian system \cref{linear ode}, produces a rational transformation,
      \begin{equation}\label{approx recursion}
          \bm x(t_{n+1}) \approx R(\tau K) \bm x(t_{n}), 
      \end{equation}
with the following properties,      
      \begin{enumerate}
          \item \cite[Lemma 6.41]{deuflhard2012scientific} The rational approximation satisfies $R(\tau K) = e^{\tau K} + O(\tau^{2p+1}).$ 
          \item \cite[Lemma 6.20]{deuflhard2012scientific}The eigenvalues of $R
          $ have modulus equal to 1.   In particular, the discrete evolution in \cref{approx recursion} is stable, thus inheriting the property in \eqref{stable K}.
          \item \cite[Lemma 6.37]{deuflhard2012scientific} The method is invariant under a similarity transformation.
          \item \cite[Theorem 1.5 and 2.1]{hairer2006geometric} The method preserves linear and quadratic physical invariants. 
          \item \cite[Theorem 4.2]{hairer2006geometric} $R(\tau K)$  is a symplectic map for any $\tau \in \mathbb{R}.$ 
          \item \cite[Theorem 3.3]{hairer2006geometric}. For general Hamiltonian systems \eqref{eq: ham-ode}, the RKG method conserves an approximate Hamiltonian $\widetilde{H}$ for all time $t$ and $\widetilde{H}- H= O(\tau^{2p})$. Furthermore, $\widetilde{H}\equiv H$ for linear Hamiltonian systems.
      \end{enumerate}
\end{theorem}

The first property establishes the error scale relative to the step size, while the second property pertains to temporal symmetry. The third property is important for Newtonian mechanics since it preserves the property of conservation of linear and angular momentum. The last property, which remains valid for all time, has proven to be the most significant advantage of symplectic integrators.

We first give examples of $1$-stage and $2$-stage Runge-Kutta Gauss Collocation method. 
\begin{example}[Implicit midpoint method ($p=1$)]
\begin{equation}
    R(z)= \frac{2+z}{2-z}. 
\end{equation}
 
\end{example}
\begin{example}[$p=2$]  In this case, the rational function is given by
    \begin{equation}
        R(z)=\frac{z^2+6z+12}{z^2-6z+12}.
    \end{equation}
Notice that this is precisely the Pad\'e approximation of $e^z$ with order $(2,2)$. 
\end{example}

In general, the function $R(z)$ for the implicit Runge-Kutta method is expressed as
\begin{equation}
    R(z) = 1+z\bm b^T (I-z\mathcal{A})^{-1}\bm1,
\end{equation}
where $\bm b^T$ and $\mathcal{A}$ are from Butcher tableau \cite{butcher1964implicit}, which lists all parameters in a Runge-Kutta method. $\bm1$ is a vector of dimension $p$ with all ones in all entries. 

For general ODEs, the $O(\tau^{2p})$ error has been proved in \cite{deuflhard2012scientific}, but without the prefactor which might depend on $p.$
We first provide a more precise error bound when applying RKG to the linear dynamics in  \cref{linear ode}.

\begin{theorem}
    The $p$-stage Runge-Kutta Gauss method produces a symplectic mapping from $\bm x_0$ to the numerical solution $\bm x_M$ at any later time step $T=t_M$, $t_M=M\tau$, and it has a global error with the following bound,
    \begin{equation}
      \norm{\bm x(T) - \bm x_M } =     O\left( \frac{T \kappa(V)^2 (\norm{K}\tau)^{2p}}{p! 2^{2p}} \right),
    \end{equation}
  where $\kappa(V)$ is the condition number of the eigenvector matrix of $K$ in \cref{diagonalization K}.  
\end{theorem}
\begin{proof}
   The first part was proven in \cite{hairer2006geometric}. This proof relies on the property of RKG methods that they preserve any quadratic invariant. For the second part, we examine the proof in \cite{deuflhard2012scientific}, and incorporate the general error bound for Gaussian quadrature. 

    We begin with the analysis of the one-step error, which for RKG, has been reduced to the quadrature error of an integral \cite{deuflhard2012scientific}, 
    \begin{equation}\label{error eq}
        \hat{\bm \varepsilon}_n\coloneqq e^{\tau K} \bm x_n - R(\tau K) \bm x_n= e^{\tau K} \int_0^{\tau} e^{-s K } \big(\bm  u'(s) - K \bm u(s) \big) ds 
    \end{equation}
    Here $\bm u(s)$ is a polynomial of degree at most $p$ that satisfies the ODE \cref{linear ode} exactly at $p$ points. The idea is that if these points are chosen based on Gaussian quadrature points, the integral can be approximated by a Gaussian quadrature formula with maximal precision. On the other hand, because of the interpolation property of $\bm u$, the quadrature approximation is zero. In general, Gaussian quadrature for approximation of an integral of $f$ over an interval of size $\tau$ follows the error bound (denoted by $E[f]$),
    \begin{equation}
     E[f]=   \frac{\tau^{2p+1} (p!)^4 }{(2p+1)! \big((2p)!\big)^2 } \norm{f^{(2p)}}= O\left(   \frac{\tau^{2p+1}  }{(2p+1)! 2^{4p} } \norm{f^{(2p)}} \right). 
    \end{equation}

    From \cref{error eq}, the function is given by $\bm f(s)=e^{-s K } \big( \bm u'(s) - K \bm u(s) \big). $ As a result, we have,
      \begin{equation}\label{bound f 2p}
         \norm{ f^{(2p)} } \leq  \sum_{q=0}^{2p}  \binom{2p}{q}\left( \norm{K}^{2p-q} \norm{\bm u^{(q+1)}} + \norm{K}^{2p-q+1} \norm{\bm u^{(q)}} \right)\kappa(V).
      \end{equation} 
    Notice that $\bm u^{(q)}=0$ for $q\ge p$ and $\|\bm u^{(q)} - \bm x^{(q)}\| = O(\tau^{q+1})$  for $q=0,\cdots,p-1$, by \cite[Lemma 1.6]{hairer2006geometric}.
    Thus from \cref{bound f 2p}, we have
    \begin{equation}
        \|f^{(2p)}\|=O\left(\kappa(V)\sum_{q=0}^{p}\binom{2p}{q}\|K\|^{2p}\right)=O\left(\kappa(V)p2^{2p}\|K\|^{2p}\right).
    \end{equation}
    This leads to an error bound given by,
    \begin{equation}
        E[f]=O\left(\frac{p\tau^{2p+1}}{(2p+1)!2^{2p}}\|K\|^{2p}\kappa(V)\right).
    \end{equation}

    With the one-step error established, we proceed to estimate the global error, using the standard trick \cite{deuflhard2012scientific}. Specifically, we first notice that the exact solution operator $e^{tK}$, follows the bound  
      \begin{equation}\label{eq: sens}
           \norm{  e^{\tau K} \bm x_n - e^{\tau K} \bm x(t_n) } \leq \norm{e^{\tau K}} \norm{\bm x_n-\bm x(t_n)}.  
      \end{equation}

We now express the global error as follows,
\[
\bm x_{n+1} - \bm x(t_{n+1}) = \underbrace{\Big( R(\tau K) \bm x_{n} -   e^{\tau K} \bm x_{n}  \Big)}_{=: \hat{\bm \varepsilon}_n}+ \Big(  e^{\tau K} \bm x_{n}  - e^{\tau K} \bm x(t_n)  \Big).
\]
The first term on the right hand side is precisely the one-step error, while the second term can be estimated from the stability of the exact solutions of the ODEs, i.e. \cref{eq: sens}. Using a discrete variation of constant formula, we find that,
\begin{equation}
    \bm x_{M} - \bm x(t_{M}) = \sum_{m=0}^M e^{m \tau K } \bm \varepsilon_{M-m}. 
 \end{equation}
Before we take the norm and apply the submultiplicative property of matrice norms, we first recall that $K$ is diagonalizable, $K= V \Lambda V^{-1},$ and $\Lambda$ only has purely imaginary eigenvalues. Therefore,
\begin{equation}
    \bm x_{M} - \bm x(t_{M}) = V \sum_{m=0}^M e^{m \tau \lambda } V^{-1} \bm \varepsilon_{M-m}   
    \Rightarrow \norm{ \bm x_{M} - \bm x(t_{M}) } = O\left( \kappa(V)^2 M \frac{ p\tau^{2p+1}}{(2p+1)!2^{2p}}\|K\|^{2p} \right). 
\end{equation}
      
    Therefore, the global error is what is stated.

\end{proof}

With this theorem, we can choose $\tau \norm{K} = O(1) $ and  
\begin{equation}\label{parameter p}
    p= \frac12  \frac{ \log \frac{T \kappa(V)  }{\epsilon} }{ \log \log \frac{T\kappa(V) }{\epsilon}}
\end{equation} to ensure that the RKG method produces a solution within error $\epsilon. $

\subsection{The quantum algorithm and its complexity}

To solve the recursive equation \cref{approx recursion}, we rewrite it into a linear system by concatenating the equations from every time step, which subsequently can be solved by the quantum linear system algorithm (QLSA). For this purpose, we define a linear operator $L$, where 
\begin{equation}\label{A-lhs}
       L = \mat{I & 0 & 0 & \cdots & 0\\
       -R & I & 0 & \cdots &0 \\
       0 & -R & I &\cdots & 0 \\
       \vdots & \vdots & \ddots & \ddots & \vdots \\
       0 & 0 & \cdots & -R & I 
       }= \sum_{m=0}^M\ketbra{m} \otimes I -  \sum_{m=1}^M\ketbra{m}{m-1} \otimes R.
   \end{equation}

As a result, all the equations in \cref{approx recursion} can be reduced into a combined linear system of equations,  
\begin{equation}\label{linear system L}
    L\ket{\bm x}=\ket{\bm b},
\end{equation}
where $\ket{\bm x}$, known as the history state,  is given by
\begin{equation} \label{vector x}
    \ket{\bm x} = \frac{1}{\sqrt{M+1}} \sum_{n=0}^M \ket{\bm x_n} \ket{t_n}, 
\end{equation}
and 
\begin{equation}
    \ket{\bm b} = \ket{\bm x_0,0,\cdots,0}.
\end{equation}
Notice that the first row of the matrix $L$ in \cref{A-lhs} enforces the initial condition. 

\begin{lemma}\label{Cond L linear}
    The condition number of $L$ follows the bound,
\begin{equation}
    \kappa(L) = O\bigl(M \kappa(V)^2 \bigr). 
\end{equation}    
\end{lemma}

\begin{proof}
   Define $P\coloneqq I-L$. Then $L=I-P$ and $P$ is $(M+1)-$nilpotent, i.e., $P^n=O,\forall n\ge M+1$. Then 
   \begin{equation}
      \| L^{-1}\|=\left\|(I-P)^{-1}\right\|=\left\|\sum_{j=0}^M P^j\right\|\le 1+\sum_{j=1}^M \|R^j\|\le 1+M\kappa(V),
   \end{equation}
   where $\kappa(V)=\|V\|\|V^{-1}\|$ and $R(\tau K)=V\Lambda(\tau K)V^{-1}$ with $\Lambda$ be a diagonal matrix.
   Furthermore,
   \begin{equation}
       \|L\|=\|I+P\|\le 1+ \|R\|\le 1+ \kappa(V).
   \end{equation} 
   Then the condition number is as desired.
\end{proof}

\begin{remark}
   The result in \cite{krovi2023improved}, which also employed QLSA, does not involve $\kappa(V)^2$. Instead, $\sup\limits_{t\in[0,T]}\norm{e^{Kt}}$ and $g:=\frac{\max_{t\in[0,T]}\norm{\bm x(t)}}{\norm{\bm x(T)}}$ appear in the complexity. In our case, $K$ is diagonalizable and thus $\norm{e^{Kt}}$ is replaced by a time-independent parameter $\kappa(V)$ by the stability property \cref{stable K}. Further, due to the time reversibility in the property  \cref{stable K}, we also have $g\leq \kappa(V)$. Compared to \cite{krovi2023improved}, the current method preserves the symplectic property. 
\end{remark}

To boost the probability of collapsing to $\bm x_M$, we incorporate the standard techniques used in \cite{berry2014high, childs2020quantum,krovi2023improved,liu2021efficient,xue2021quantum}. We repeat $\bm x_M$ for additional $r$ times, following \cref{History boosted}. In practice, one can choose $r=O(M)$ to achieve a finite probability.

\begin{lemma}\label{History boosted}
    If we extend \cref{vector x} to 
    \begin{equation}
        \ket{\bm x} = \frac{1}{\sqrt{M+1+r}}\left(\sum_{n=0}^{M}\ket{\bm x_n}\ket{t_n}+\sum_{n=M+1}^{M+r}\ket{\bm x_M}\ket{t_n}\right),
    \end{equation}
    by adding extra r rows of zeros in the $\ket{\bm b}$,
    then $\ket{\bm x}$ satisfies a corresponding linear system of equations $\tilde{L}\ket{\bm x}=\ket{\bm b}$ with 
    \begin{equation}
        \kappa(\tilde{L})=O\left((M+r)\kappa(V)^2\right).
    \end{equation}
\end{lemma}
\begin{proof}
    Notice that
    \[ \tilde{L} = \mat{L & 0 \\ 0 & B}, \]
    where $B$ is a $r\times r$ block matrix with
    \[ B = \mat{I & 0 & 0 & \cdots & 0\\
       -I & I & 0 & \cdots &0 \\
       0 & -I & I &\cdots & 0 \\
       \vdots & \vdots & \ddots & \ddots & \vdots \\
       0 & 0 & \cdots & -I & I 
       }. \]
    We directly get $\|\tilde{L}\|=O(\kappa(V))$. To calculate $\|\tilde{L}^{-1}\|$, we use similar techniques for $\|L^{-1}\|$. Let $\tilde{L}=I-\tilde{P}$, then $\tilde{P}$ is $(M+r+1)$-nilpotent. For any $j=1,\cdots,M+r$, we know $\tilde{P}^j$ only has one nonzero element at each column and is in the form of $R^k$ with $k$ at most $j$. Collecting these steps, one arrives at 
    \begin{equation}
        \|\tilde{L}^{-1}\|=\left\|\frac{1}{I-\tilde{P}}\right\| \le 1+\sum_{i=1}^{M+r}\|R^i\|=O((M+r)\kappa(V)).
    \end{equation}
    This concludes the result.
\end{proof}

An important step in the implementation of the QLSA procedure is the preparation of the matrix $R$. Following the RKG method, each step of the solution method can be rewritten as \cite{deuflhard2012scientific}
\begin{equation}
\left\{
\begin{aligned}
        \bm k_i = & K \bm x_n + \tau \sum_{j=1}^p a_{ij} K \bm k_j, \quad \forall i \in [p] \\
        \bm x_{n+1} = & \bm x_n  + \tau \sum_{i=1}^p b_i \bm k_i.
\end{aligned} \right.
\end{equation}
The parameters $\{b_j, a_{i,j}\}$ are determined from Gaussian quadrature.  Let $\mathcal{A}\in \mathbb{R}^{p\times p}$ and $\bm b\in \mathbb{R}^{p}$ denote those two sets of parameters, respectively. 
We can deduce the matrix $R$ by direct substitution,
\begin{equation}\label{R-mat}
    R(\tau K)= I +  (\bm b^T \otimes I_n)  G^{-1}  (\mathbf{1} \otimes (\tau K)), \quad G:=I_p \times I_n - \mathcal{A} \otimes (\tau K).
\end{equation}
Here $\mathbf{1} = \sum_{i=1}^p \ket{i} $ is a $p$-dimensional vector with all entries being one. 

Now the problem is reduced to the computation of the inverse of the matrix $G$. 

\begin{theorem}\label{thm: Construct R}
  Suppose that $\tau \norm{K} < 1,$ then the matrix $G$ in \cref{R-mat} is invertible.
  Further, if  $\tau \norm{K} <  \frac{1}{2\sqrt{p}},$ then $\kappa(G) < 2 + 2\sqrt{p}$,  
  and thus the matrix $R$ in \cref{R-mat} can be prepared using $\tilde{O}(1)$ queries to $K$ with gate complexity $\tilde{O}(s)$,  where in $\Tilde{O}$ we neglected logarithmic factors
\end{theorem}
\begin{proof}

Implicit Runge-Kutta methods that are derived from collocation methods are extensively treated in \cite{wanner1996solving}. In particular, for the RKG method, the parameter matrix  $\mathcal{A}$ is similar to a tri-diagonal matrix \cite[Theorem 5.6]{wanner1996solving},
\[
\begin{bmatrix}
    \frac12 & -\frac{1}{2\sqrt{3}} &  0 & 0 & \cdots &  0 &  0\\
     \frac{1}{2\sqrt{3}} & 0 & -\frac{1}{2\sqrt{15}}   & 0 & \cdots  &  0& 0 \\
   \vdots &  \vdots & \ddots & \ddots  & \ddots & \vdots &\vdots  \\
    0 & 0  &0  & 0& \cdots & 0  & - \frac{1}{2\sqrt{4(p-1)^2-1}}  \\
   0 & 0  &0  & 0& \cdots &\frac{1}{2\sqrt{4(p-1)^2-1}}  & 0
\end{bmatrix}
\]
Therefore, the spectral radius $\varrho(\mathcal{A})< \frac{1}{2}(1+ \frac{1}{\sqrt{3}})<1$, 
and in light of \cref{R-mat}, $G$ is invertible.

We proceed to estimate the condition number of $G$, 
\[
    \norm{G} \leq 1 + \norm{\mathcal{A}} \leq 1 +  \sqrt{p}\norm{\mathcal{A}}_{\infty} \leq 
    1 + \sqrt{p}. 
    \]
Here we first use the property of the matrix norms, followed by the result  \cite{caklovic2022infinity} $\norm{\mathcal{A}}_{\infty} \leq \sqrt{\xi_p}<1$, where $\xi_p$ is the largest root of the Legendre polynomial of degree $p$.

Meanwhile, due to the fact that $\varrho(\mathcal{A})<1$ and $\tau \norm{K} < 1,$ we can use the Neumann series to express the inverse of $G$ 
\[
G^{-1} = \sum_{m\geq 0} \mathcal{A}^m \otimes (\tau K)^m \Rightarrow \norm{G^{-1}} \leq 
\frac{1}{1- \norm{\mathcal{A}}_\infty/2 } \leq 2.
\]
The last step used the fact that $\norm{\mathcal{A}\otimes (\tau K)} \leq \norm{\mathcal{A}}\frac{1}{2\sqrt{p}}<1/2.$

With these properties established for the matrix $G$, we can construct a block encoding of $G^{-1}$. Specifically, recall that $K$ is $s-$sparse from \cref{sparsity K}, implying that $G$ is sparse with sparsity at most $(ps+1)-$ from its definition in \cref{R-mat}. Note that from \cref{parameter p}, $p$ is introduces a logarithmic factor, implying that the sparsity is still $\tilde{O}(s)$.    Next, following the procedure in \cite[Lemma 2 and Remark 1]{krovi2023improved}, we first define $\overline{G}$ as the Hermitian complement of $G$, i.e., according to \cref{Hermitian Complement}, which has the same condition number as $G$. Therefore,  \cref{lem:BE-inv} guarantees that  
we can construct a $(2\kappa(G), a+3+O(\log(\kappa(G)^2\log\epsilon^{-1})),\epsilon)$ block encoding of $\overline{G}^{-1}$ with gate complexity
\begin{equation}
  O\left(ps\kappa(G)(a+3+O(a+\log^{2.5}(4s^2/\delta)))\log^2(\kappa(G)^2/\epsilon)\right) = \tilde{O}(s),
\end{equation}
where $a$ is qubit-size of matrix $\overline{G}$. As suggested in \cite[Remark 1]{krovi2023improved}, a block encoding of $G^{-1}$ can be extracted afterwards. Therefore, we can obtain from \cite[Theorem 4.1]{camps2024explicit} and the estimate on $p$ in \cref{parameter p} the block encoding of $R$ with $\tilde{O}(1)$ queries to $K,$ and gate complexity $\tilde{O}(s)$. Here we assume that the access to the parameters in the RKG methods are given. 
    
\end{proof}

Now we arrive at the first main theorem in this paper. 

\begin{theorem}
    The quantum algorithm using QLSA produces a quantum state $\ket{\psi} = \frac{\bm x_M}{\norm{\bm x_M}}   $, satisfying $\norm{\bm x_M - \bm x(t)}< \epsilon$. The mapping from $\bm x_0$ to $\bm x_M$ is symplectic. The algorithm involves a query complexity  $ \Tilde{O}\bigl(T \|K\| \kappa(V)^2\bigr)$ of $K$. 
\end{theorem}
\begin{proof}
The symplectic property follows directly from that of the RKG method. With the structure of $L$ in \cref{A-lhs} and \cref{thm: Construct R}, we can block encode $L$ with gate complexity $\tilde{O}(s)$ from \cref{lem:BE-lcu}. Recall that $\kappa(L)$ is bounded by $O(M\kappa(V)^2)$ in \cref{Cond L linear}. By \cref{complexity lemma}, we know the query complexity is $O(\kappa(L)\log\frac{1}{\varepsilon})(T_0+T_{in})$, where $T_0$ is the complexity of block encoding $L$ and $T_{in}$ is the complexity of preparing the initial state. Note that $T_0=\tilde{O}(1)$, and $T_{in}=O(1)$ according to \cite[Lemma 14]{krovi2023improved}. Therefore, the total query complexity in constructing $L$ requires $\tilde{O}(\kappa(L))=\tilde{O}(M\kappa(V)^2)$ queries to $K$. Recall that $\tau M = T$ and $\tau\norm{K}=O(1)$. Thus, $\tilde{O}(M\kappa(V)^2)=\tilde{O}(T\norm{K}\kappa(V)^2)$. 
\end{proof}

\section{Nonlinear Hamiltonian systems}

We now generalize our algorithm to nonlinear Hamiltonian systems. In particular, we consider a Hamiltonian  that can be expanded as a polynomial,
\begin{equation}
    H(\bm x)= H_2 \bm x^{\otimes 2} + H_3 \bm x^{\otimes 3}+ H_4 \bm x^{\otimes 4} + \cdots 
\end{equation}
The type of Hamiltonians that can be expanded in this form has been extensively studied in the theory of dynamical systems, especially for their generic solution structure \cite{cabral2023normal}. 

To keep our discussions simple, we focus on the case where the Hamiltonian only consists of the first two terms, and the general case can be handled similarly. We also write $Q=\frac12 H_2$ and $C=\frac13 H_3.$ Without loss of generality, we can impose the following symmetry: $Q_{i,j}$ and $C_{i,j,k}$ are invariant under permutations of the indices. Therefore, the nonlinear dynamical system can be written as,
\begin{equation}\label{nonlinearHam}
    \frac{d}{dt} \bm x= J^{-1} Q \bm x +J^{-1} C \bm x \otimes \bm x.
\end{equation}

\begin{example}
  Consider the FPU chain model for a system of  $L$ particles with quadratic and cubic terms in the potential energy,
    \begin{equation}
        U(z)=\frac{k}{2}z^2+\frac{\alpha}{3}z^3. 
    \end{equation}
    Then the Hamiltonian $H(\bm p, \bm q)$  is given by, 
    \begin{equation}
        H =  \sum_{i=1}^L \frac{p_i^2}{2m} + \sum_{i=1}^{L-1}\left(\frac{k}{2}(q_{i+1}-q_i)^2+\frac{\alpha}{3}(q_{i+1}-q_i)^3 \right).
    \end{equation}
\end{example}

\subsection{Quantum algorithms via Carleman embedding}

Carleman embedding is a common technique for transforming a nonlinear ODE system into an infinite-dimensional linear ODE system. We refer readers to the related works \cite{amini2022carleman, brustle2024quantum,liu2021efficient, wu2024quantum}. We express the ODEs in \cref{nonlinearHam} in a form considered in \cite{liu2021efficient}
\begin{equation}
    \frac{d}{dt}\bm x = F_1\bm x+F_2\bm x^{\otimes2},
\end{equation}
where for Hamiltonian system in \cref{nonlinearHam}, we choose $F_1 = -JQ$ and $F_2=-JC$. In Carleman embedding, one defines a new set of variables, 
$\bm y = \mat{\bm x & \bm x^{\otimes2} &\bm x^{\otimes3} & \cdots}^T$, which from \cref{nonlinearHam}, follow an \emph{infinite} linear ODE system from Carleman,
\begin{equation}\label{eq: car}
    \frac{d}{dt}\bm y = A \bm y,
\end{equation}
where the structure of the matrix $A$ is same as  those  in \cite{liu2021efficient,wu2024quantum}.

Here we introduce the same assumption for the no-resonance condition in our previous work \cite{wu2024quantum}.

\begin{assumption}\label{assump1}
  Let  $\{\l_1,\cdots,\l_n\}$  be the eigenvalues of $F_1$, i.e.,
  \begin{equation}\label{F1}
      F_1 = V \Lambda V^{-1}.
  \end{equation}
  We assume that, for any $i\in[n]$,
    \begin{equation}\label{cond-res}
        \l_i \ne \sum_{j=1}^n m_j\l_j \quad \forall m_j\in \mathbb{Z} \; \text{and} \; m_j \ge0 \text{ s.t. } \sum_{j=1}^n m_j\ge2. 
    \end{equation}
Equivalently, there exists a constant $\Delta >0,$ such that
\begin{equation}\label{Delta}
    \norm{ \left((\l_{m_1}+\l_{m_2}+\cdots+\l_{m_n})I-F_1\right)^{-1} } \leq \frac{1}{\Delta}.
\end{equation}
\end{assumption}

Notice that the diagonalizability is consistent with our assumption in \cref{as-stab} for the linear Hamiltonian systems. Meanwhile, the no-resonance condition ensures the nonlinear stability \cite{wu2024quantum}.

\paragraph{Properties of the infinite Carleman system. }
Suppose that the solution of the nonlinear dynamics has a uniform in time bound,
\( \norm{\bm x(t)}< \mu \), in the time interval of interest,  we can define a Banach space similar to that in \cite{amini2021error},
\begin{equation}
    {\ell}^2:= \{ \bm y: \sum_{j\geq 1} \mu^{-2j} \norm{\bm y_j}^2 < \infty \}.
\end{equation}

\begin{theorem}\label{inf}
  Assume that the solution of the nonlinear ODE satisfies a uniform bound,
    \begin{equation}
        \norm{\bm x(t)}_1 \leq \mu, \; \forall t\in [0,T],
    \end{equation}
 $F_1$ fulfills \cref{assump1}, and that 
 \begin{equation}\label{Rr}
   \mathsf{R}_r <1, \quad   \mathsf{R}_r:= \frac{ 4e\mu \kappa_1(V)\norm{F_2}_1}{\Delta},
 \end{equation}
where $\kappa_1(V)=\norm{V}_1\norm{V^{-1}}_1$.
Then the infinite Carleman system has a unique solution $\bm y(t) \in {\ell}^2$, and $\bm y_j = \bm x(t)^{\otimes j}, \forall j\in \mathbb{N}.$
\end{theorem}
\begin{proof}
This can be proved by taking  a finite truncation of the Carleman system  \eqref{eq: car}, then taking the limit as the truncation level goes to infinite. See \cref{car-convg} and \cite[Theorem 2.12]{wu2024quantum}.
\end{proof}

To discuss the symplectic property, we must examine the fundamental matrix $W,$ $W:=\frac{\partial \bm y(t) }{\partial \bm y(0)}$. We first show that the first diagonal block is symplectic. 
\begin{lemma}
    Under the same condition as in \cref{inf}, the matrix $W_{11}:= \frac{\partial \bm y_1(t) }{\partial \bm y_1(0)} $ associated with the solution of the infinite Carleman system \eqref{eq: car} is symplectic.
\end{lemma}
\begin{proof}
    To fulfill this property, it requires $ \frac{d}{dt}  W_{11}^T J W_{11}=0, $ i.e.,
\[
W_{11}^T A_{11}^T J W_{11} + W_{11}^T J A_{11} W_{11}+
 W_{21}^T A_{12}^T J W_{11}  + W_{11}^T J A_{12} W_{21}=0.
 \]
In particular, $A_{11} = J^{-1} Q $ with $Q$ being a symmetric matrix. Therefore, the first two terms are canceled, and we are led to the condition,
\[
 W_{21}^T A_{12}^T J W_{11}  + W_{11}^T J A_{12} W_{21}=0.
\]
Note that \[ W_{21}(t) = \frac{\partial \bm y_2(t)}{\partial \bm y_1(0)}= W_{11}\otimes \bm x + \bm x \otimes W_{11}. \]
and 
\begin{equation}
    W_{21}^T A_{12}^TJW_{11} = (A_{12}W_{21})^TJW_{11}=\left(J^{-1}CW_{21}\right)^TJW_{11}=W_{21}^TC^T(J^{-1})^TJW_{11}=-W_{21}^TC^TW_{11}.
\end{equation}
Therefore, we see the equivalence, 
\[ W_{21}^T A_{12}^T J W_{11}  + W_{11}^T J A_{12} W_{21}=0\iff \left(CW_{21}\right)^TW_{11}=W_{11}^T\left(CW_{21}\right). \]
It remains to prove that $\left(CW_{21}\right)^TW_{11}$ is symmetric: 
\begin{equation}
    \begin{aligned}
        \left(CW_{21}\right)^TW_{11} &= \left(W_{11}\otimes \bm x+\bm x\otimes W_{11}\right)^TC^TW_{11}\\
        &= \left(W_{11}^T\otimes\bm x^T+\bm x^T\otimes W_{11}\right)C^TW_{11}\\
        &= W_{11}^T C \left(W_{11}^T\otimes\bm x^T+\bm x^T\otimes W_{11}\right)^T\\
        &= W_{11}^T C W_{21}.
    \end{aligned}
\end{equation}
The equality follows from the invariant of $C$ under permutation. We conclude that $W_{11}^T J W_{11}$ equals to some constant matrix. Since $W_{11}(0)=I$, we have
\[ W_{11}^T J W_{11}=J,  \]
implying $W_{11}$ is sympletic.
\end{proof}

Another interesting perspective of the Carleman embedding is that it induces a nonlinear coordinate transformation that maps the nonlinear ODEs into another simpler ODE. Under the no-resonance condition, the new ODE system becomes linear \cite{tsiligiannis1989normal}, which in the theory of dynamical systems \cite{katok1995introduction}, is known as the normal form.  An elegant result from Hamiltonian systems is that they can be mapped to another Hamiltonian system, and the mapping is also a symplectic transformation \cite{cabral2023normal}. For Carleman embedding, we can state this connection as follows,

\begin{theorem}
    Assume the convergence of the Carleman embedding and the absence of resonance. Then there exists an invertible coordinate transformation $\bm z_1 = \bm h(\bm x)$, 
such that
\begin{equation}\label{eq: z_1}
    \frac{d}{dt} \bm z_1 = \Lambda \bm z_1, \quad  \Lambda:= \mathrm{diag} \{ \lambda_1, \lambda_2, \cdots, \lambda_n \}.
\end{equation} 

\end{theorem}

\begin{proof}

Let us write the initial condition $\bm y(0)$ as $\bm y(0)=\phi(\bm x(0))$ that maps $\bm x$ to all the tensor-powers of $\bm x$. Since $A$ is diagonalizable in our case, $A= W D W^{-1}$ with $W$ and $W^{-1}$ being both upper triangular. Thus
\begin{equation}\label{y2z}
    \bm z= W^{-1} \bm y
\end{equation}
 satisfies decoupled ODEs.  

In particular, the first block, written as $\bm z_1 = E_1 \bm z$, with $E_1=[I \; 0\; \cdots \; 0 \; \cdots ]$, satisfies \cref{eq: z_1}.
Therefore, this nonlinear map transforms the nonlinear ODE into a linear one. This mapping is called a near-identity transformation.  Now we determine the map $\bm z_1= \bm h(\bm x).$ By following the above procedure, we find that,
\begin{equation}
    \bm h(\bm x)= E_1 W^{-1} \phi(\bm x) = \sum_{j\geq 1} Z_{1j} \bm x^{\otimes j}.  
\end{equation} 
Here $Z= W^{-1}$.

In order to complete this normal form interpretation, we construct $ \bm h^{-1}(\bm z_1)$ from the Carleman embedding.
We start by examining the remaining block in \cref{y2z},
\begin{equation}\label{xtoz_k}
    \bm z_k = \sum_{\ell \geq k} W_{k,\ell} \bm x^{\otimes \ell}. 
\end{equation}
Meanwhile, from the case $k=1$, we have,
\begin{equation}\label{xtoz^k}
    \bm z_1^{\otimes k} = \sum_{\ell \geq k} R_{k,\ell} \bm x^{\otimes \ell}, 
\end{equation}
where we have introduced a matrix $R$
\begin{equation}
    R_{k,\ell} = \sum_{ \overset{ j_1, j_2, \cdots, j_k \geq 0}{ j_1+ j_2+ \cdots+ j_k = \ell}} Z_{1,j_1}  Z_{1,j_2} \cdots  Z_{1,j_k}.  
\end{equation}

Since $Z$ is invertible,  the matrix $R$ is invertible as well. We can define $T= Q R^{-1}$, and we find that,
\begin{equation}
     \bm z_k = \sum_{\ell \geq k} T_{k,\ell} \bm z_1^{\otimes \ell}. 
\end{equation}

Finally, using the connection from \cref{y2z}, we find that,
\begin{equation}
    \bm x =\sum_{k\geq 1} W_{1,k} \bm z_k =  \sum_{\ell \geq 1} S_{\ell} \bm z_1^{\otimes \ell} = \bm h^{-1} (\bm z_1), \quad 
    S_{\ell} = \sum_{k\geq 1}  Z_{1,k}  T_{k,\ell}.  
\end{equation}

A simply way to verify that this mapping is $\bm h^{-1}$, i.e., 
\begin{equation}\label{hhinv}
    \bm x= \bm h^{-1}\big( \bm h(\bm x) \big),
\end{equation}
is to express the above mapping in a matrix form,
\begin{equation}
    \bm x= E_1 W Z R^{-1} \bm \zeta=E_1 R^{-1} \bm \zeta, \quad \bm \zeta= (\bm z_1, \bm z_1^{\otimes 2}, \bm z_1^{\otimes 3}, \cdots ). 
\end{equation}
Meanwhile, \cref{xtoz^k} implies that $\bm \zeta = R \bm y$, which confirms \cref{hhinv}.

\end{proof}

\paragraph{Properties of a finite Carleman system. } Our quantum algorithm will be applied to a \emph{finite} truncation of the infinite Carleman system \eqref{eq: car}.
We now estimate the impact on the symplectic property due to the truncation. Toward this end, let ${\bm y}^{(N)}$ be the solution of the truncated system from  \eqref{eq: car} by keeping the first $N$ blocks in $\bm y$. 
Similarly, we define $A^{(N)}$ as the corresponding truncated matrix. That is, after $N$-th level of truncation for Carleman embedding, we have
\begin{equation}\label{N-Carleman}
    \frac{d}{dt}\bm y^{(N)} = A^{(N)} \bm y^{(N)}.
\end{equation}
Here the structure of the matrix is as follows,
\begin{equation}\label{matrix-car}
A^{(N)}=    \begin{pmatrix}
A_{1,1}  & A_{1,2} & 0& 0 & \cdots & 0 & 0 \\
0  & A_{2,2} & A_{2,3} &  0& \cdots & 0 & 0 \\
 0 & 0 & A_{3,3} &  A_{3,4}& \ddots & \vdots & \vdots \\
 0 & 0 & &A_{4,4}  & \ddots & \vdots & \vdots  \\
  \vdots  & \vdots  &  \vdots &\vdots  & \ddots & \vdots & \vdots  \\
 0 & 0 & 0 & 0 & \cdots & A_{N-1,N-1} & A_{N-1,N}  \\
 0 & 0 & 0 & 0& \cdots  & 0 & 
A_{N,N}\end{pmatrix},
\end{equation}
where the blocks are defined as,
\begin{equation}\label{Ajj}
    \begin{aligned}
        A_{j,j}=&\sum_{i=1}^{j-1}I^{\otimes i}\otimes F_{1}\otimes I^{\otimes (j-1-i)} \in \mathbb{R}^{n^j \times n^j},\; 
        A_{j,j+1}=&\sum_{i=1}^{j-1}I^{\otimes i}\otimes F_{2}\otimes I^{\otimes (j-1-i)} \in \mathbb{R}^{n^j \times n^{j+1}}.
    \end{aligned}
\end{equation}
$I$ refers to the identity matrix acting on $\mathbb{R}^n.$

\begin{theorem}\cite[Theorem 2.12]{wu2024quantum}\label{car-convg}
   Under the same assumptions as  \cref{inf}, the error in the Carleman embedding can be bounded by, 
\begin{equation}
    \norm{\bm x(T) -\bm  y_1(T) }_2 \leq C N  T  \mathsf{R}_r^{N-1}. 
\end{equation}
\end{theorem}
We see from this theorem that the convergence of the Carleman embedding does not significantly depend on the evolution time $T$. To find an explicit bound, we show that,
\begin{corollary}\label{cor:N}
    The error due to the Carleman embedding can be controlled to be within $\epsilon$, with $\epsilon < \frac{T} {e \log \frac{1}{\mathsf{R}_r} }$, provided that 
    \begin{equation}\label{chooseN}
        N \in \Omega \left( \frac{\log \frac{T\log \frac{1}{\mathsf{R}_r}}{\epsilon}}{\log \frac{1}{\mathsf{R}_r}} \right).
    \end{equation}
\end{corollary}
\begin{proof}
    Following \cref{car-convg}, we may define $y= N \log \mathsf{R}_r$, and rewrite the bound as $y e^y \approx \frac{\epsilon \log \mathsf{R}_r} {T}$. When $\epsilon < \frac{T} {e \log \frac{1}{\mathsf{R}_r} }$, this can be associated with the Lambert-W function $W_{-1}.$ Therefore, the bound on $N$ follows the inequality of $W_{-1}$ \cite{chatzigeorgiou2013bounds}. 
\end{proof}
This result shows that the truncation level only has a logarithmic dependence on $T$ and $\epsilon$, and therefore does not significantly increase the overall complexity.

\medskip  
 
Recall that our focus is on the symplectic properties of the solutions. Toward that end, we define  $W_{11}^{(N)}$ as the Wronskian matrix, i.e., $W_{11}^{(N)} = \frac{\partial}{\partial \bm x_0} {\bm y}_1^{(N)}$.  We consider the error  $\bm \eta  $,
\begin{equation}
    \bm \eta(t) = \begin{bmatrix}
        W_{11} \\
        W_{21}\\
        \vdots \\ 
        W_{N1}
    \end{bmatrix} -   \frac{\partial }{ \partial \bm x_0}  {\bm y}^{(N)}.
\end{equation}

 With the truncation and the linearity of the equation \eqref{N-Carleman}, we get an error equation
\begin{equation}
     \frac{d}{dt} \bm \eta(t)   = A^{(N)} \bm \eta(t)  +  
     \begin{bmatrix}
         0\\  0 \\  \vdots \\ 0 \\ A_{N,N+1} \frac{\partial \bm y(t)^{\otimes (N+1)} }{\partial \bm y_1(0)} 
     \end{bmatrix}, \quad \bm \eta(0)=0.
\end{equation}

This has the same structure as the error equation in the analysis of Carleman embedding \cite{wu2024quantum}.
Therefore, by following the same proof as in the previous theorem, we arrive at an approximate symplectic property, 
\begin{theorem}
 Under the assumption that $R_r<1$,  and $\frac{\partial \bm y(t) }{\partial \bm y_1(0)}$ is bounded, then
 the Jacobian matrix $W_{11}^{(N)}$ is approximately symplectic, in the sense that
 \begin{equation}\label{symp-e}
      \norm{ \left(W_{11}^{(N)}\right)^T J  W_{11}^{(N)} - J } \leq \epsilon,
      \end{equation}
 provided that $N$ is selected based on \cref{chooseN}.
\end{theorem}

\subsection{Quantum Algorithm for the Carleman system}

Now we can turn to the finite Carleman system in \cref{N-Carleman} and discuss the quantum algorithms to solve the linear ODE system.
Similar to the previous section, we still use block encoding as input models to access the matrices $F_1$ and $F_2$
in the matrix $A$ in the fintie Carleman system in \cref{N-Carleman}. In light of the matrix structure in \cref{matrix-car}, we assume that $F_1$ and $F_2$ are sparse, which thanks to \cref{lemma:sparse-to-be}, can be efficiently block encoded.

Similar to the previous section, we apply the RKG method to the finite Carleman system \cref{N-Carleman}, 
\begin{equation}
    \bm y_{m+1} = R\left(\tau A^{(N)} \right) \bm y_{m},
 \end{equation}
 where with slight abuse of notations, we use $\bm y_m$ to denote the approximate solution of $\bm y^{(N)}$ at time $t_m$; $t_m= m \tau.$
Similar to the QLSA approach in our previous section, this produces 
a linear system $L \ket{\bm y} = \ket{\bm b}$ as well, where
\begin{equation}\label{L-lhs}
       L = \mat{I & 0 & 0 & \cdots & 0\\
       -R & I & 0 & \cdots &0 \\
       0 & -R & I &\cdots & 0 \\
       \vdots & \vdots & \ddots & \ddots & \vdots \\
       0 & 0 & \cdots & -R & I 
       } = \sum_{m=0}^M\ketbra{m} \otimes I -  \sum_{m=1}^M\ketbra{m}{m-1} \otimes R,
   \end{equation}
Where the matrix $R$ stands for the rational function $R\left(\tau A^{(N)} \right)$ now.
Then \cref{approx recursion} can be encoded into $L\ket{\bm y}=\ket{\bm b}$, where the history state $\ket{\bm y}$ is
\begin{equation} \label{vector y}
    \ket{\bm y} = \frac{1}{\sqrt{M+1}} \sum_{n=0}^M \ket{\bm y_n} \ket{t_n}, 
\end{equation}
and 
\begin{equation}
    \ket{\bm b} = \ket{\bm y_0,0,\cdots,0}.
\end{equation}
For the remaining section, we provide an estimation for the condition number $\kappa(L)$.

\begin{lemma} \label{nonlinear: norm R}
    If $T\le \frac{1}{e\norm{F_2}\kappa(V)^N}$, then the rational approximation of $e^{tA^{(N)}}$, $R$, follows the bound,
    \begin{equation}
        \norm{R} = \tilde{O}\left(\left(\frac{\epsilon}{T}\right)^{-\log\kappa(V)}\right).
    \end{equation}
\end{lemma}
Now we have the bound for $R$, we follow same steps asin \cref{Cond L linear}, we derive the bound for $\kappa(L)$.

\begin{lemma} \label{nonlinear: cond L}
    Under the same assumption as \cref{nonlinear: norm R}, the condition number of $L$, $\kappa(L)$, follows the bound,
    \begin{equation}
        \kappa({L}) = \tilde{O}\left(M\left(\frac{\e}{T}\right)^{-2\log\kappa(V)}\right).
    \end{equation}
\end{lemma}
Proof of \cref{nonlinear: norm R} and \cref{nonlinear: cond L} are in \cref{proof R} and \cref{proof L}.

We can follow the steps in \cref{thm: Construct R} to construct $R$ and $L$ for the  quantum circuit. Now we can now apply the same analysis as in the quantum linear solver approach for solving ODEs and arrive at the second main result of this paper.

\begin{theorem}\label{thm2}
    Under the same assumption as \cref{nonlinear: norm R}, consider the finite Carleman system truncated at level $N$ with $N$ following the choice \cref{chooseN} and the condition \cref{Rr}.  
    The quantum algorithm combining the RKG method and QLSA produces an approximate $\ket{\psi} \propto \bm x_M $, with $\norm{\bm x_M - \bm x(T) }< \epsilon$ and $\bm x(T)$ being the solution of the nonlinear Hamiltonian system \eqref{nonlinearHam}. The algorithm incurs, aside from logarithmic factors, $\Tilde{O}\left(\frac{T^{1+2\log \kappa(V)}}{\e^{2\log\kappa(V)}}\Big(\norm{F_1}+\norm{F_2}\Big) \right)$ queries to coefficient matrices in  \eqref{nonlinearHam}. The mapping from $\bm x(0)$ to $\bm x(T)$ is $\epsilon$-symplectic in the sense of \cref{symp-e}.
\end{theorem}
\begin{proof}
    By using \cref{lem:qls} and \cref{nonlinear: cond L}, we get the result directly.
\end{proof}

If time $T$ is beyond the assumption in \cref{nonlinear: norm R}, a direct calculation gives us \[\kappa(L)=\tilde{O}\left(M\left(\frac{\e}{T}\right)^{-2\log(\kappa(V)+eT\norm{F_2})+2\log\frac{\e}{T}\log\kappa(V)}\right).\] The query complexity is \[\tilde{O}\left(\frac{T^{1+2\log(\kappa(V)+eT\norm{F_2})-2\log\frac{\e}{T}\log\kappa(V)}}{\e^{2\log(\kappa(V)+eT\norm{F_2})-2\log\frac{\e}{T}\log\kappa(V)}}(\norm{F_1}+\norm{F_2})\right). \]

\section{Summary and Discussions}
In this paper, we present quantum algorithms for simulating Hamiltonian systems described by ordinary differential equations (ODEs). Our emphasis is on maintaining the symplectic structure of the solutions, which is crucial for preserving energy conservation and thus ensuring accuracy in long-time simulations. For linear Hamiltonian systems, we can exactly retain this property by incorporating symplectic integrators that we select to also have optimal accuracy. For nonlinear Hamiltonian systems, we employ the Carleman embedding technique to embed the nonlinear dynamics into a larger, but linear, system, which can also be handled by symplectic integrators. We provide a condition that guarantees the accuracy in approximating the nonlinear dynamics. Importantly, we also show that the symplectic property is retained within the same precision.

However, it remains an open issue whether there exist quantum algorithms that can exactly preserve the symplectic property for nonlinear dynamics. Moreover, the bound for the condition in \cref{nonlinear: norm R} involves $\kappa(V)^N$, exhibiting an exponential dependence on the Carleman truncation level~$N$, which seems difficult to circumvent. For instance, examining the last block of~$\bm y$, whose dynamics is completely driven by~$A_{N,N}$ with eigenvectors in the matrix~$V^{\otimes N}$, reveals this exponential dependence. Such exponential dependence is absent in the work of Krovi~\cite{krovi2023improved}, owing to the dissipative properties of~$A_{N,N}$ inherited from~$F_1$, which do not hold for Hamiltonian systems.

Recent work by Brüstle et al.~\cite{brustle2024quantum} established a lower bound for simulating nonlinear unitary dynamics, demonstrating that the complexity scales exponentially with the evolution time~$T$. Dynamics that reach this lower bound are clearly outside the regime considered in this paper. However, extending our current techniques and analysis to algorithms based on short-time Carleman embedding presents an interesting direction for future research.

\section{Acknowledgement}
This research is supported by the NSF Grant DMS-2411120.

%\bibliography{cite,ham,QCdynamicalsystem}
%\bibliographystyle{plain}

\begin{appendix}

\section{Proof of \cref{nonlinear: norm R}} \label{proof R}
Since $R$ is the approximation of $e^{tA^{(N)}}$, we may directly estimate the norm bound of $e^{tA^{(N)}}$.
Let $A^{(N)}=A_0+A_1$, where $A_0$ is the block diagonal matrix and $A_1$ is an upper block off-diagonal matrix. Let $C(t)= e^{tA^{(N)}}$. We have
\begin{equation}
    \frac{d}{dt}C(t) = A_0C(t)+A_1C(t).
\end{equation}
By using variation of constant formula, we obtain
\begin{equation}\label{vocf}
    \begin{aligned}
        C(t)&=e^{tA_0}+\int_0^t e^{(t-t_1)A_0}A_1C(t_1)dt_1\\
        &=e^{tA_0}+\int_0^t e^{(t-t_1)A_0}A_1e^{t(A_0+A_1)}dt_1\\
        &=e^{tA_0}+\int_0^t e^{(t-t_1)A_0}A_1\left(e^{t_1A_0}+\int_0^{t_1}e^{(t_1-t_2)A_0}A_1e^{t_2A}\right)dt_1\\
        &=e^{tA_0}+\int_0^t e^{(t-t_1)A_0}A_1e^{t_1A_0}dt_1+\int_0^t\int_0^{t_1}e^{(t-t_1)A_0}A_1e^{(t_1-t_2)A_0}A_1e^{t_2A_0}dt_2dt_1+\cdots.
    \end{aligned}
\end{equation}
Here we use the fact that for any matrices $A=B+C$,
\begin{equation}
    e^{tA} = e^{tB}+\int_0^t e^{(t-s)B}Ce^{sA}ds.
\end{equation}
Notice that $A_1$ is $N$-nilpotent. That is, $A_1^k=0$ for all $k\ge N$. Furthermore, $A_0$ is block diagonal, $e^{tA_0}$ is also block diagonal. Thus, the product of $e^{tA_0}$ and $A_1$ does not change the nilpotent structure of $A_1$. Therefore, \cref{vocf} only contains $N$ terms,
\begin{equation}
\begin{aligned}
    C(t)&=e^{tA_0}+\int_0^t e^{(t-t_1)A_0}A_1e^{t_1A_0}dt_1+\cdots\\
    &\qquad\qquad\qquad+\int_0^t\int_0^{t_1}\cdots\int_0^{t_{N-2}}e^{(t-t_1)A_0}A_1e^{(t_1-t_2)A_0}A_1\cdots e^{(t_{N-2}-t_{N-1})A_0}A_1e^{t_{N-1}A_0}dt_{N-1}\cdots dt_1dt.
\end{aligned}
\end{equation}
Furthermore, $\norm{e^{tA_0}}\le\kappa(V)^N$ and $\norm{A_1}\le N\norm{F_2}$. This leads to the inequality,
\begin{equation}
    \norm{C(t)}\le \kappa(V)^N+t\kappa(V)^{2N}N\norm{F_1}+\cdots+\frac{t^{N-1}}{(N-1)!}\kappa(V)^{N\cdot N}(N\norm{F_2})^{N-1} = \kappa(V)^N\sum_{k=0}^{N-1}\frac{t^k}{k!}\left(\kappa(V)^N N\norm{F_2}\right)^k.
\end{equation}
From \cref{chooseN}, we know $\kappa(V)^N = \tilde{O}\left(\left(\frac{\epsilon}{T}\right)^{-\log\kappa(V)}\right)$. Since $T\le \frac{1}{e\norm{F_2}\kappa(V)^N}$ and by Stirling's approximation, we conclude that
\begin{equation}
    \norm{C(t)} \le N\kappa(V)^N = \tilde{O}\left(\left(\frac{\epsilon}{T}\right)^{-\log\kappa(V)}\right).
\end{equation}

\section{Proof of \cref{nonlinear: cond L}} \label{proof L}
From the proof of \cref{Cond L linear}, $\norm{L}=O\left(\norm{R}\right)$ and $\norm{L^{-1}}=O\left(\sum_{j=1}^M \norm{R^j}\right)$. Note that $j\tau \le T$ for all $j\in[M]$, we have 
\begin{equation}
    \norm{R^j} = \tilde{O}\left(\left(\frac{\epsilon}{T}\right)^{-\log\kappa(V)}\right), \forall j\in[M].
\end{equation}
Therefore, $\norm{L^{-1}}=\tilde{O}\left(M\left(\frac{\epsilon}{T}\right)^{-\log\kappa(V)}\right)$ and we conclude the proof.

\end{appendix}

\end{document}